\documentclass[11pt]{article}

\usepackage{amsmath,amssymb,amsthm,amsbsy}
\usepackage{mathrsfs} 
\usepackage{graphicx}	
\usepackage{fullpage}	
\usepackage{multirow} 
\usepackage{lineno}
\usepackage{algorithm}
\usepackage{subfigure}
\usepackage{enumerate}	

\usepackage[bookmarks=true, linkbordercolor={1 1 1}, citebordercolor={1 1 1}]{hyperref} 

\theoremstyle{plain} \newtheorem{lemma}{Lemma}
\theoremstyle{plain} \newtheorem{proposition}{Proposition}
\theoremstyle{plain} 
\theoremstyle{definition} 
\theoremstyle{definition} \newtheorem{definition}{Definition}
\theoremstyle{definition} \newtheorem{assumption}{Assumption}
\theoremstyle{definition} 

\theoremstyle{plain} \newtheorem*{lemma*}{Lemma}
\theoremstyle{plain} \newtheorem*{proposition*}{Proposition}
\theoremstyle{plain} \newtheorem*{theorem*}{Theorem}
\theoremstyle{definition} \newtheorem*{corollary*}{Corollary}
\theoremstyle{definition} \newtheorem*{definition*}{Definition}
\theoremstyle{definition} \newtheorem*{assumption*}{Assumption}
\theoremstyle{definition} \newtheorem*{example*}{Example}

\newcommand{\R}{\ensuremath{\mathbb{R}}}										
 
\newcommand{\argmax}{\operatornamewithlimits{argmax}}
\newcommand{\argmin}{\operatornamewithlimits{argmin}}
\newcommand{\maximize}{\operatornamewithlimits{maximize}}

\newcommand{\E}[1]{\mathbb{E}\left[#1\right]}								
\newcommand{\Prob}[1]{\mathbb{P}\left[#1\right]}						
\newcommand{\indicator}[1]{\textbf{1}_{\left\{#1\right\}}} 	

\newcommand{\mb}[1]{\ensuremath{\mathbf{#1}}}
\newcommand{\bs}[1]{\boldsymbol #1}													
\newcommand{\mc}[1]{\ensuremath{\mathcal{#1}}}							

\newcommand{\beq}[1]{\begin{equation} \label{eq:#1}}
\newcommand{\eeq}{\end{equation}}
\newcommand{\beqn}{\begin{equation*}}
\newcommand{\eeqn}{\end{equation*}}

\newcommand{\setN}{\ensuremath{\mathcal{N}}}	
\newcommand{\setS}{\ensuremath{\mathcal{S}}}	
\newcommand{\setB}{\ensuremath{\mathcal{B}}}	
\newcommand{\setX}{\ensuremath{\mathcal{X}}}	
\newcommand{\setA}{\ensuremath{\mathscr{A}}}	
\newcommand{\setG}{\ensuremath{\mathcal{G}}} 	
\newcommand{\setE}{\ensuremath{\mathcal{E}}}	

\newcommand{\pb}[1]{p(#1)}
\newcommand{\pv}[2]{p(#1 | #2)}

\title{Revenue Optimal Auction for Single-Minded Buyers}

\author{
Vineet Abhishek\thanks{Department of Electrical and Computer Engineering, University of Illinois at Urbana-Champaign. \newline Contact: {\tt \small abhishe1@illinois.edu}}
~and Bruce Hajek\thanks{Department of Electrical and Computer Engineering, University of Illinois at Urbana-Champaign. \newline Contact: {\tt \small b-hajek@illinois.edu}
\newline The authors gratefully acknowledge several helpful conversations with Steven R. Williams, Department of Economics, University of Illinois at Urbana-Champaign.}
}

\begin{document}
\date{September 12, 2010}
\maketitle

\begin{abstract}
We study the problem of characterizing revenue optimal auctions for single-minded buyers. Each buyer is interested only in a specific bundle of items and has a value for the same. Both his bundle and its value are his private information. The bundles that buyers are interested in and their corresponding values are assumed to be realized from known probability distributions independent across the buyers. We identify revenue optimal auctions with a simple structure, if the conditional distribution of any buyer's valuation is nondecreasing, in the hazard rates ordering of probability distributions, as a function of the bundle the buyer is interested in. The revenue optimal auction is given by the solution of a maximum weight independent set problem. We provide a novel graphical construction of the weights and highlight important properties of the resulting auction.
\end{abstract}

\section{Introduction} \label{sec:introduction}
Consider $N$ buyers competing for a certain set of items offered by a seller. A buyer has a \textit{value} for each combination of the items (henceforth referred to as a \textit{bundle}) that he is interested in. This is the maximum price that he is willing to pay for the bundle. The seller's objective is to maximize his revenue from the sale. If the seller knew these values exactly, he could maximize his revenue by simply finding an allocation of the items among the buyers that has the maximum total value, and charging each buyer the exact value of the bundle allocated to him. However, the values that the buyers have for different bundles are their private information; the seller has imperfect information about these values. The seller's task is complicated further by the strategic behavior of the buyers; a buyer might misreport the values of the bundles he is interested in, if it is beneficial for him to do so. \textit{Combinatorial auctions} (henceforth CAs) offer a solution. CAs allow buyers the freedom to compete for any bundle of items and provide them incentives to report their private values truthfully. The allocation and the payments are determined by the competition among the buyers. However, the inherent problems of CAs limit their appeal. 

Often there are \textit{complementarities} among the items - a buyer can have a higher value for a bundle as a whole than the sum of values of the parts of the bundle. Different buyers may have different forms of complementarity. Because of complementarity, allocation of items in a CA requires solving a hard combinatorial optimization problem; a naive implementation can lead to the \textit{exposure problem} for the winners - a winner might end up getting only a part of his desired bundle, while still paying a high price for it. Moreover, the goal of \textit{revenue maximization} (referred to as \textit{optimality}) is very different from the much studied goal of social welfare maximization (referred to as \textit{efficiency}) in auction theory, e.g., the VCG mechanism \cite{Clarke71,Groves73,Vickrey61}. Theoretical results on revenue maximization are known only under simple settings, e.g., \cite{Myerson81,Hartline&Karlin07,Armstrong2000}. Characterizing revenue optimal CAs in full generality with multiple buyers, multiple items, and different values for the different bundles, still remains a challenging open problem, even if the complexity considerations are ignored.

This paper aims to address the issues of (1) dealing with complementarity, and (2) maximizing revenue, for CAs. The key to designing CAs with desirable properties is to understand and handle complementarity among the items. An extreme case of complementarity is when buyers are \textit{single minded}~\cite{Lehmann02}. Here, each buyer is interested only in a specific bundle and has a value for the same. Any allocation of items to a buyer that does not contain his desired bundle has zero value for him. Both the bundle that a buyer is interested in and its value are his private information. CA problems with single-minded buyers are more tractable than if buyers have general valuations. While the single-minded buyers model might appear as a simplifying assumption, no general result on revenue maximization is known even for this extreme case. Here a buyer has two dimensions for misreporting his preference - the bundle he is interested in, and the value of the bundle. Most of the existing literature on revenue optimal auctions studies problems which are one dimensional - a buyer has only one real number for misreporting his preference, e.g., Myerson's single item auction \cite{Myerson81}, the \textit{single-parameter} buyers model described in~\cite{Hartline&Karlin07}, and single-minded buyers with \textit{known} bundles\footnote{Single-minded buyers with \textit{known} bundles is also an instance of the single-parameter buyers.}~\cite{Ledyard07}. Thus, the single-minded buyer model can be thought of as an initial step towards solving the general CA problems. Also, see \cite{Lehmann02} for some real examples where buyers are single minded. Hence, we focus on the CAs with single-minded buyers.

In addition, we take a departure from the continuous variable models of economics and assume that the set of possible values that a buyer can have for his bundle is finite. While working with finite valuation sets is clearly relevant from the implementation point of view, it is interesting in its own right; e.g., we show later in this paper that the revenue equivalence result of~\cite{Myerson81} does not hold true when the valuation sets are finite. A related work was done by \cite{Elkind07} where a Bayesian optimal auction, when buyers' valuation sets are finite, is characterized. However, \cite{Elkind07} deals only with single item auctions.

We make the following contributions in this paper. We modify and extend the framework of~\cite{Elkind07} for multiple item auctions with single-minded buyers. We then find a sufficient condition under which a revenue optimal auction can be characterized for single-minded buyers. This sufficient condition is the monotonicity of the conditional distribution of any buyer's valuation, in the \textit{hazard rate ordering}, as a function of the bundle the buyer is interested in. An interpretation of this condition is as follows: if there are two bundles where one contains the other, then the larger bundle is likely to have a higher value. Such monotonicity property is intuitive for single-minded buyers. Such monotonicity property is intuitive for single-minded buyers. A similar hazard rate ordering condition appears in \cite{Dizdar09} in the context of dynamic knapsack problem. However, \cite{Dizdar09} (as well as \cite{Ledyard07}) assumes that the distributions also satisfy Myerson's nonintuitive regularity assumption \cite{Myerson81}. An important contribution of our paper is to show that such assumption is unnecessary. We present an algorithm for optimal auction as a solution of a maximum weight independent set problem, where weights are an appropriate mapping of buyers' valuations, called \textit{virtual valuations}. We also provide a novel graphical construction of the virtual valuations.

The rest of the paper is organized as follows. Section~\ref{sec:model} outlines our model, definitions, and notation. Section \ref{sec:opt-auction} characterizes an optimal auction for single-minded buyers, while Section \ref{sec:discussion} describes some of its important properties. We conclude in Section \ref{sec:conclusions}.

\section{Model and Notation} \label{sec:model}
Consider $N$ buyers competing for $S$ items that a seller wants to sell. The set of buyers is denoted by $\setN \triangleq \{1,2,\ldots,N\}$, and the set of items for sale is denoted by~$\setS$. Buyers are single minded - each buyer~$n$ is interested only in a specific bundle~$b_n^{*} \in 2^{\setS}$ and has a value $v_n^{*}$ for any bundle~$b_n$ such that $b_n \supseteq b_n^{*}$, while he has zero value for any other bundle. Here $2^{\setS}$ denotes the power set of~$\setS$. Notice that single-minded buyers enjoy \textit{free disposal} of the items. We refer to the tuple $(b_n^{*},v_n^{*})$ as the \textit{type} of buyer $n$. The type of a buyer is known only to him and constitutes his private information.

For each buyer $n$, the seller and the other buyers have imperfect information about his type; they describe the bundle that buyer $n$ is interested in by a random set $B_n$, and its value by a discrete random variable $X_n$. The random set~$B_n$ takes values from the collection $\setB_n \subseteq 2^{\setS}$, where $\setB_n$ is the collection of all bundles that buyer $n$ can possibly be interested in. The random variable $X_n$ is assumed to take values from the set $\mc{X}_n \triangleq \{x_n^1, x_n^2,\ldots,x_n^{K_n}\}$ of cardinality $K_n$, where $0 \leq x_n^1 < x_n^2 < \ldots < x_n^{K_n}$. The joint probability distribution of $B_n$ and $X_n$ is common knowledge. The probability that $B_n = b_n$ is denoted by $\pb{b_n}$, and conditioned on $B_n = b_n$, the probability that $X_n$ is equal to $x_n^i$ is denoted by $\pv{x_n^i}{b_n}$; i.e., $\pb{b_n} \triangleq \Prob{B_n = b_n}$, and $\pv{x_n^i}{b_n} \triangleq \Prob{X_n = x_n^i | B_n = b_n}$. Assume that $\pb{b_n} > 0$ and $\pv{x_n^i}{b_n} > 0$ for all $n \in \setN$, $b_n \in \setB_n$, and $1 \leq i \leq K_n$. Note that $(b_n^{*},v_n^{*})$ can be interpreted as a specific realization of the random variables $(B_n, X_n)$. Let $Y_n \triangleq (B_n, X_n)$ be the random vector describing the type of buyer $n$. Random vectors $[Y_n]_{n \in \setN}$ are assumed to be independent\footnote{This is referred to as the independent private value model. It is a fairly standard model in auction theory.}.

Denote a typical \textit{reported type} (henceforth referred to as a \textit{bid}) of a buyer $n$ by $(b_n,v_n)$, where $b_n \in \setB_n$ and $v_n \in \mc{X}_n$. Define the vector of bids as $(\mb{b},\mb{v})$, where $\mb{b} \triangleq (b_1, b_2, \ldots, b_N)$ is the vector of reported bundles, and $\mb{v} \triangleq (v_1, v_2, \ldots, v_N)$ is the vector of reported values. The seller can only allocate the items to a set of buyers whose reported bundles are disjoint. Given a vector of bundles~$\mb{b}$, define $\mc{A}(\mb{b})$ as follows:
\beq{feasible_allocations}
\mc{A}(\mb{b}) \triangleq \{A \subseteq \setN : ~ \forall n,m ~\in A, ~ n \neq m, ~ b_n \cap b_m = \emptyset \}. 
\eeq
This is the collection of all feasible allocations; i.e., the collection of all subsets of buyers who can be allocated their respective bundles simultaneously. Trivially, $\emptyset \in \mc{A}(\mb{b})$ and $\mc{A}(\mb{b})$ is \textit{downward closed}; i.e., if $A \in \mc{A}(\mb{b})$ and $B \subseteq A$, then $B \in \mc{A}(\mb{b})$.

Define $\mb{B} \triangleq (B_1,B_2,\ldots,B_N)$, $\mb{X} \triangleq (X_1, X_2, \ldots, X_N)$, and $\mb{Y} \triangleq (Y_1, Y_2, \ldots, Y_N)$. We use $\mb{Y}$ and $(\mb{B},\mb{X})$ interchangeably. Let $\bs{\setB} \triangleq \setB_1 \times \setB_2 \times \ldots \times \setB_N$ and $\bs{\setX} \triangleq \setX_1 \times \setX_2 \times \ldots \times \setX_N$. We use the standard game theoretic notation of $\mb{v}_{-n} \triangleq (v_1, \ldots, v_{n-1}, v_{n+1}, \ldots, v_N)$ and $\mb{v} \triangleq (v_n, \mb{v}_{-n})$. Similar interpretations are used for $\mb{b}_{-n}$, $\mb{B}_{-n}$, $\mb{X}_{-n}$, $\mb{Y}_{-n}$, $\bs{\setB}_{-n}$, and $\bs{\setX}_{-n}$. Henceforth, in any further usage, $b_n$, $\mb{b}_{-n}$, and $\mb{b}$ are always in the sets $\setB_n$, $\bs{\setB}_{-n}$, and $\bs{\setB}$ respectively; and $v_n$, $\mb{v}_{-n}$, and $\mb{v}$ are always in the sets $\setX_n$, $\bs{\setX}_{-n}$, and $\bs{\setX}$ respectively.

\section{Revenue Optimal Auction} \label{sec:opt-auction}
In this section, we formally describe the optimal auction problem, formulate the objective and the constraints explicitly, and provide an optimal algorithm for solving the problem. We will be focusing only on the auction mechanisms where buyers are asked to report their types directly (referred to as \textit{direct mechanism}). By the \textit{revelation principle}\footnote{Revelation principle says that, given a mechanism and a Bayes-Nash equilibrium (BNE) for that mechanism, there exists a direct mechanism in which truth-telling is a BNE, and allocation and payment outcomes are same as in the given BNE of the original mechanism.} \cite{Myerson81}, the restriction to direct mechanisms is without any loss of optimality.

\subsection{Characterization} \label{sec:characterization}
An auction mechanism is specified by an allocation rule $\bs{\pi}: \bs{\setB} \times \bs{\setX} \mapsto [0,1]^{2^N}$, and a payment rule $\mb{M}:\bs{\setB} \times \bs{\setX} \mapsto~\R^N$. Given a bid vector $(\mb{b},\mb{v})$, the allocation rule $\bs{\pi}(\mb{b},\mb{v}) \triangleq [\pi_A(\mb{b},\mb{v})]_{A \in 2^{\setN}}$ is a probability distribution over the power set $2^{\setN}$ of $\setN$. For each $A \in 2^{\setN}$, $\pi_A(\mb{b},\mb{v})$ is the probability that the set of buyers $A$ get their reported bundles simultaneously. The payment rule is defined as $\mb{M} \triangleq (M_1,M_2,\ldots,M_N)$, where $M_n(\mb{b},\mb{v})$ is the payment (expected payment in case of random allocation) that buyer~$n$ makes to the seller when the bid vector is $(\mb{b},\mb{v})$. Let $Q_n(\mb{b},\mb{v})$ be the probability that buyer $n$ gets his reported bundle $b_n$ when the bid vector is $(\mb{b},\mb{v})$; i.e,
\beq{winning-prob}
Q_n(\mb{b},\mb{v}) \triangleq \sum_{A \in 2^{\setN} : n \in A} \pi_A(\mb{b},\mb{v}).
\eeq

Given that the type of a buyer $n$ is $(b_n^{*},v_n^{*})$, and the bid vector is $(\mb{b},\mb{v})$, the payoff (expected payoff in case of random allocation) of the buyer $n$ is:
\beq{payoff}
\sigma_n(\mb{b},\mb{v}; b_n^{*},v_n^{*}) \triangleq Q_n(\mb{b},\mb{v})\indicator{b_n^{*} \subseteq b_n}v_n^{*} - M_n(\mb{b},\mb{v}).
\eeq
So buyers are assumed to be risk neutral and have quasilinear payoffs (a standard assumption in auction theory). The mechanism $(\bs{\pi},\mb{M})$ and the payoff functions $[\sigma_n]_{n \in \setN}$ induce a game of incomplete information among the buyers. We use Bayes-Nash equilibrium (BNE) as the solution concept. The seller's goal is to design an auction mechanism $(\bs{\pi},\mb{M})$ to maximize his expected revenue at a BNE of the induced game. Again, using the revelation principle, seller can restrict only to the auctions where truth-telling is a BNE (referred to as \textit{incentive compatibility}) without any loss of optimality.

For the above revenue maximization problem to be well defined, assume that the seller cannot force the buyers to participate in an auction and impose arbitrarily high payments on them. Thus, a buyer will voluntarily participate in an auction only if his payoff from participation is nonnegative (referred to as \textit{individual rationality}). In addition, the auction mechanism that the seller uses must always produce feasible allocations; i.e., for any bid vector, the set of winners must have disjoint bundles. The seller too is assumed to have free disposal of the items and may decide not to sell some or all items for certain bid vectors.

The idea now, as in \cite{Myerson81}, is to express incentive compatibility, individual rationality, and feasible allocations as mathematical constraints, and formulate the revenue maximization objective as an optimization problem under these constraints. To this end, for each $n \in \setN$, $b_n$, and $v_n$, define the following functions:
\beq{q}
q_n(b_n,v_n) \triangleq \E{Q_n(b_n,v_n,\mb{Y}_{-n})},
\eeq
\beq{m}
m_n(b_n,v_n) \triangleq \E{M_n(b_n,v_n,\mb{Y}_{-n})},
\eeq
Here, $q_n(b_n,v_n)$ is the expected probability that buyer $n$ gets his bundle given that he reports his type as $(b_n,v_n)$ while everyone else is truthful. The expectation here is over the type of everyone else; i.e., over $\mb{Y}_{-n}$. Similarly, $m_n(b_n,v_n)$ is the expected payment that buyer $n$ makes to the seller. The constraints can be expressed mathematically as follows
\begin{enumerate}
\item
\textit{Feasible allocation (FA)}: For any $\mb{b}$ and $\mb{v}$,
\beq{fa}
A \notin \setA(\mb{b}) \Rightarrow \pi_A(\mb{b},\mb{v}) = 0.
\eeq

\item
\textit{Incentive compatibility (IC)}: For any $n \in \setN$, $b_n$, $t \in \setB_n$, and $1 \leq i, j \leq K_n$, 
\beq{ic}
q_n(b_n,x_n^i)x_n^i - m_n(b_n,x_n^i) \geq q_n(t,x_n^j)\indicator{b_n \subseteq t}x_n^i - m_n(t,x_n^j).
\eeq
Notice that, given $B_n= b_n$, and $X_n = x_n^i$, the left side of \eqref{eq:ic} is the payoff of buyer $n$ from reporting his type truthfully, assuming everyone else is also truthful, while the right side is the payoff from misreporting his type to $(t,x_n^j)$. 

\item
\textit{Individual rationality (IR)}: For any $n \in \setN$, $b_n$, and $1 \leq i \leq K_n$,
\beq{ir}
q_n(b_n,x_n^i)x_n^i - m_n(b_n,x_n^i) \geq 0.
\end{equation}
\end{enumerate}

Under IC, all buyers report their true types. Hence, the expected revenue that the seller gets is $\E{\sum_{n=1}^{N}M_n(\mb{Y})}$. The expectation here is over the distribution of the random vector $\mb{Y}$. Thus, the seller's optimization problem is given by: \newline

\noindent \textbf{Optimal auction problem (OAP)}
\beq{oap}
\begin{array}{l}
\displaystyle \maximize_{\bs{\pi},\mb{M}} \quad \E{\sum_{n=1}^{N}M_n(\mb{Y})}, \\
\text{subject to FA, IC, and IR constraints.}
\end{array}
\eeq

Instead of solving the OAP, we solve a modified problem obtained by relaxing the IC constraint. We then find a sufficient condition under which the solution of the modified problem is also the solution of OAP. The \textit{relaxed IC constraint} is obtained by assuming that buyers report their bundles truthfully, or equivalently, the bundles that the buyers are interested in are known to everyone. Mathematically, the relaxed IC constraint is:
\beq{relaxed-ic}
q_n(b_n,x_n^i)x_n^i - m_n(b_n,x_n^i) \geq q_n(b_n,x_n^j)x_n^i - m_n(b_n,x_n^j),
\eeq
for any $n \in \setN$, $b_n$, and $1 \leq i, j \leq K_n$. The modified optimization problem is given by: \newline

\noindent \textbf{Modified optimal auction problem (MOAP)}
\beq{moap}
\begin{array}{l}
\displaystyle \maximize_{\bs{\pi},\mb{M}} \quad \E{\sum_{n=1}^{N}M_n(\mb{Y})}, \\
\text{subject to FA, relaxed IC \eqref{eq:relaxed-ic}, and IR constraints.}
\end{array}
\eeq

For the MOAP, using the relaxed IC constraint \eqref{eq:relaxed-ic} and the IR constraint \eqref{eq:ir}, we relate the expected payment $m_n$ of a buyer $n$ to his expected allocation probability $q_n$. The framework used is similar to \cite{Elkind07}. We first define the \textit{virtual-valuation} function, $w_n$, of buyer $n$ as:
\beq{virtual-bid}
w_n(b_n,x_n^i) \triangleq \left\{ 
\begin{array}{l l}
  \displaystyle x_n^i - (x_n^{i+1} - x_n^i)\frac{\left(\sum_{j=i+1}^{K_n}\pv{x_n^j}{b_n}\right)}{\pv{x_n^i}{b_n}} & \quad \text{if $1 \leq i\leq K_n -1$,}\\
  x_n^{K_n} & \quad \text{if $i = K_n$.}\\
\end{array} \right.
\eeq

\begin{definition} \label{definition:regularity}
The virtual-valuation function $w_n$ is said to be regular if $w_n(b_n,x_n^i) \leq w_n(b_n,x_n^{i+1})$ for all $b_n$, and $1 \leq i \leq K_n-1$.
\end{definition}

\begin{proposition}
\label{proposition:opt-revenue}
Let $\bs{\pi}$ be an allocation rule and $[Q_n]_{n \in \setN}$ and $[q_n]_{n \in \setN}$ be obtained from $\bs{\pi}$ by \eqref{eq:winning-prob} and \eqref{eq:q}. A payment rule satisfying the relaxed IC constraint \eqref{eq:relaxed-ic} and the IR constraint \eqref{eq:ir} exists for $\bs{\pi}$ if and only if $q_n(b_n,x_n^i) \leq q_n(b_n,x_n^{i+1})$ for all $n \in \setN$, $b_n$, and $1\leq i \leq K_n-1$. Given such $\bs{\pi}$ and a payment rule $\mb{M}$ satisfying the IC and IR constraints, the seller's revenue satisfies: 
\beq{revenue}
\text{Seller's revenue } = \E{\sum_{n=1}^{N}M_n(\mb{Y})} \leq \E{\sum_{n=1}^{N}Q_n(\mb{Y})w_n(Y_n)}.
\eeq
Moreover, a payment rule $\mb{M}$ achieving this bound exists and any such $\mb{M}$ satisfies: 
\beqn
m_n(b_n,x_n^i) = \sum_{j = 1}^{i}(q_n(b_n,x_n^j) - q_n(b_n,x_n^{j-1}))x_n^j,
\eeqn
for all $n \in \setN$, $b_n$, and $1\leq i \leq K_n$, where we use the notational convention $q_n(b_n,x_n^{0}) \triangleq 0$.
\end{proposition}
\begin{proof}
The proof is given in Appendix \ref{sec:appendixA}.
\end{proof}

\subsection{Solution of the MOAP} \label{sec:moap-solution}
We now describe an algorithm for finding a solution of the MOAP. As mentioned in Section \ref{sec:introduction}, this is related to \cite{Hartline&Karlin07}, \cite{Ledyard07}, and \cite{Elkind07}.

From \eqref{eq:winning-prob}, for all $b_n$ and $v_n$, we have:
\beq{pi-q}
\sum_{n=1}^{N}Q_n(\mb{b},\mb{v})w_n(b_n,v_n) = \sum_{A \in 2^{\setN}}\pi_A(\mb{b},\mb{v})\bigg(\sum_{n \in A}w_n(b_n,v_n)\bigg).
\eeq
Proposition \ref{proposition:opt-revenue} and \eqref{eq:pi-q} suggest that a solution of the MOAP can be found by selecting the allocation rule $\bs{\pi}$ that assigns nonzero probabilities only to the set of buyers in $\mc{A}(\mb{b})$ with the maximum total virtual valuations for each bid vector $(\mb{b},\mb{v})$. If all $w_n$'s are regular, then it can be verified that such an allocation rule satisfies the monotonicity condition on the $q_n$'s needed by Proposition~\ref{proposition:opt-revenue}. However, if $w_n$'s are not regular, the resulting allocation rule would not necessarily satisfy the required monotonicity condition on the $q_n$'s. This problem can be remedied by using another function, $\overline{w}_n$, called the \textit{monotone virtual valuation} (henceforth MVV), constructed graphically as follows.

For all $n \in \setN$, $b_n$, and $0 \leq i \leq K_n$, define:
\beq{gh}
(g_n^{b_n,i},h_n^{b_n,i}) \triangleq \bigg(\sum_{j = 1}^{i} \pv{x_n^j}{b_n}, -x_n^{i+1}\big(\sum_{j = i+1}^{K_n} \pv{x_n^j}{b_n}\big)\bigg),
\eeq
where we use the notational convention of $\sum_{j=1}^{0}(.) \triangleq 0$, $x_n^{K_n+1} \triangleq 0$, and $\sum_{j=K_n+1}^{K_n}(.) \triangleq 0$. Then, $w_n(b_n,x_n^i)$ is given by the slope of the line joining the point $(g_n^{b_n,i-1},h_n^{b_n,i-1})$ to the point $(g_n^{b_n,i},h_n^{b_n,i})$; i.e., 
\beq{h-slope}
w_n(b_n,x_n^i) = \frac{h_n^{b_n,i}-h_n^{b_n,i-1}}{g_n^{b_n,i}-g_n^{b_n,i-1}}. 
\eeq
Find the lower convex hull of the points $[(g_n^{b_n,i},h_n^{b_n,i})]_{0 \leq i \leq K_n}$. Let $\overline{h}_n^{b_n,i}$ be the point on this convex hull corresponding to $g_n^{b_n,i}$. Then, $\overline{w}_n(b_n,x_n^i)$ is defined as the slope of the line joining the point $(g_n^{b_n,i-1},\overline{h}_n^{b_n,i-1})$ to the point $(g_n^{b_n,i},\overline{h}_n^{b_n,i})$;~i.e., 
\beq{h-slope-monotone}
\overline{w}_n(b_n,x_n^i) = \frac{\overline{h}_n^{b_n,i}-\overline{h}_n^{b_n,i-1}}{g_n^{b_n,i}-g_n^{b_n,i-1}}. 
\eeq
To find the points $[(g_n^{b_n,i},h_n^{b_n,i})]_{0 \leq i \leq K_n}$ graphically, draw vertical lines separated from each other by distances $\pv{x_n^1}{b_n}, \pv{x_n^2}{b_n}, \ldots, \pv{x_n^{K_n}}{b_n}$. For each $1 \leq i \leq K_n$, join the point $-x^i$ on the y-axis to the x-axis at $1$ (sum of probabilities) and call such line as line $i$. The intersection of line $1$ with y-axis is the point $(g_n^{b_n,0},h_n^{b_n,0})$. The intersection of line $2$ with the first vertical line is the point $(g_n^{b_n,1},h_n^{b_n,1})$. Similarly, the intersection of the line $3$ with the second vertical line is the point $(g_n^{b_n,2},h_n^{b_n,2})$, and so on. Notice that if $w_n$ is regular, $\overline{w}_n$ is equal to~$w_n$. 

Figure \ref{fig:virtual-bids} shows this construction for a typical random variable $X$ taking four different values $\{x^1,x^2,x^3,x^4\}$ with corresponding probabilities $\{p^1,p^2,p^3,p^4\}$, where we have dropped the subscripts corresponding to the buyers and the bundle information for the ease of notation. This is the case of virtual valuation being regular. Since the slopes of the graph are nondecreasing, the function graphed is convex. In Figure \ref{fig:monotone-bids}, $w(x^1) > w(x^2)$, and hence, the virtual-valuation function is not regular. Here, the the lower convex hull of the points $(g^i,h^i)$'s is taken. The slopes of individual segments of this convex hull give the MVV function $\overline{w}$. This is equivalent to replacing $w(x^1)$ and $w(x^2)$ by their weighted mean; i.e., $\overline{w}(x^1) = \overline{w}(x^2) = (p^1w(x^1) + p^2w(x^2))/(p^1 + p^2)$. 

\begin{figure}[htp]
\centering
\subfigure[Virtual valuations as the slopes of the graph.]{
\includegraphics[trim=0.8in 0.8in 0.8in 0.8in, clip=true, height=5.35in, angle=270]{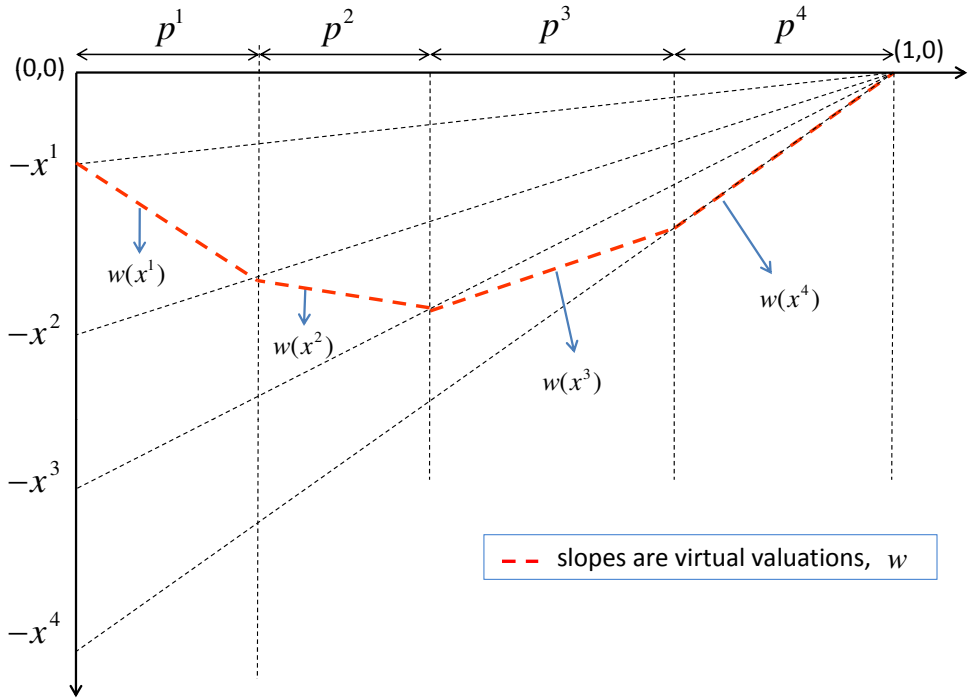}
\label{fig:virtual-bids}
}
\subfigure[Virtual valuations and MVVs as the slopes of the graph.]{
\includegraphics[trim=0.8in 0.8in 0.8in 0.8in, clip=true, height=5.35in, angle=270]{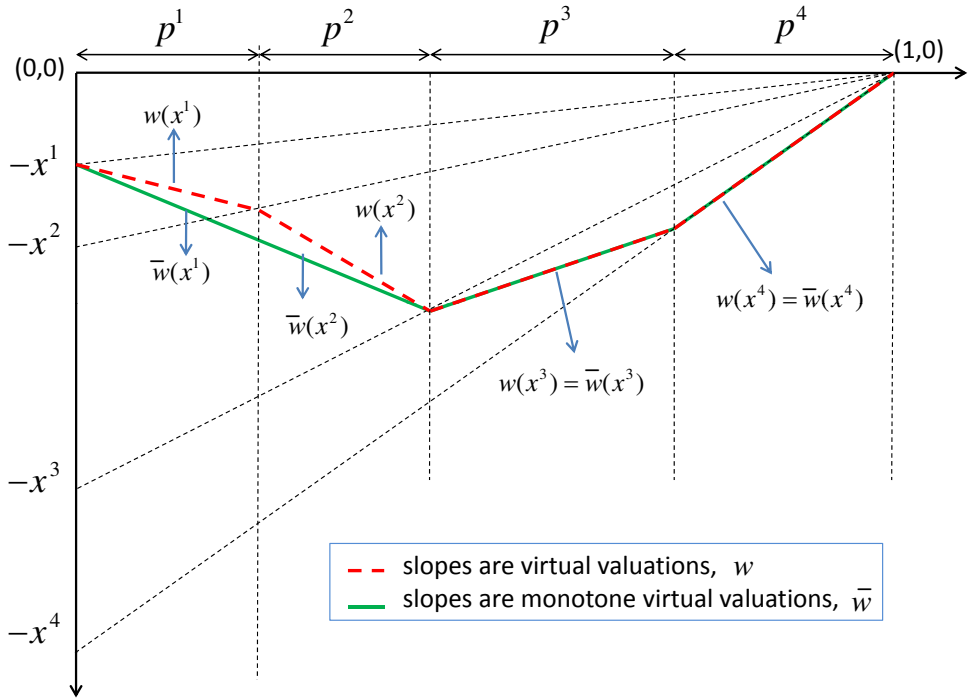}
\label{fig:monotone-bids}
}
\caption{Graphical construction of virtual valuations and MVVs.}
\label{fig:vv-mvvs}
\end{figure}

The following lemma that is a straightforward consequence of the construction of $\overline{w}_n$ as the slopes of a convex function:
\begin{lemma}
\label{lemma:monotone-vv-lemma}
$\overline{w}_n(b_n,x_n^i) \leq \overline{w}_n(b_n,x_n^{i+1})$ for all $n \in \setN$, $b_n$, and $1 \leq i \leq K_n-1$.
\end{lemma}

The next proposition establishes the significance of the allocation rule obtained by using $\overline{w}_n$'s.

\begin{proposition} \label{proposition:monotone-vv-opt}
Given any allocation rule $\bs{\pi}$ such that $[Q_n]_{n \in \setN}$ and $[q_n]_{n \in \setN}$, obtained from $\bs{\pi}$ by \eqref{eq:winning-prob} and \eqref{eq:q}, satisfy $q_n(b_n,x_n^i) \leq q_n(b_n,x_n^{i+1})$, for all $n \in \setN$, $b_n$, and $1 \leq i \leq K_n-1$. Then, 
\beq{monotone_bid_ineq}
\E{\sum_{n=1}^{N}Q_n(\mb{Y})w_n(Y_n)} \leq \E{\sum_{n=1}^{N}Q_n(\mb{Y})\overline{w}_n(Y_n)}.
\eeq
Moreover, \eqref{eq:monotone_bid_ineq} holds with equality for any allocation rule that maximizes $\sum_{n=1}^{N}Q_n(\mb{b},\mb{v})\overline{w}_n(b_n,v_n)$ for each bid vector $(\mb{b},\mb{v})$, subject to the FA constraint.
\end{proposition}
\begin{proof}
The proof appears in Appendix \ref{sec:appendixB}. The proof is a straightforward extension of \cite{Elkind07} where a similar result is obtained for single item auctions.
\end{proof}

The maximum weight algorithm (henceforth MWA) for the MOAP is described in Algorithm~\ref{alg:mwa}. The set~$\mc{W}(\mb{b},\mb{v})$ is the collection of all feasible subsets of buyers with maximum total MVVs for the given bid vector $(\mb{b},\mb{v})$. Since $\setA(\mb{b})$ is downward closed and $\emptyset \in \setA(\mb{b})$, no buyer $n$ with $\overline{w}_n(b_n,v_n) < 0$ is included in the set of winners $W(\mb{b},\mb{v})$. In step $3$ of the MWA, for each $\text{$x_n^i \leq v_n$}$, $Q_n(\mb{b},x_n^i,\mb{v}_{-n})$ is computed recursively by treating $(\mb{b},x_n^i,\mb{v}_{-n})$ as the input bid vector and repeating steps $1-2$.

\begin{algorithm}
\floatname{algorithm}{Algorithm}
\caption{Maximum weight algorithm (MWA)} \label{alg:mwa}
Given a bid vector $(\mb{b},\mb{v})$:
\begin{enumerate}
\item
Compute $\overline{w}_n(b_n,v_n)$ for each $n \in \setN$.

\item
Take $\bs{\pi}(\mb{b},\mb{v})$ to be any probability distribution on the collection $\mc{W}(\mb{b},\mb{v})$ defined as:
\beqn 
\mc{W}(\mb{b},\mb{v}) \triangleq \argmax_{A \in \setA(\mb{b})} \sum_{n \in A}\overline{w}_n(b_n,v_n).
\eeqn
Obtain the set of winners $W(\mb{b},\mb{v})$ by sampling from $\mc{W}(\mb{b},\mb{v})$ according to $\bs{\pi}(\mb{b},\mb{v})$.

\item
Collect payments given by:
\beqn
M_n(\mb{b},\mb{v}) = \sum_{i:x_n^i \leq v_n}\big(Q_n(\mb{b},x_n^i,\mb{v}_{-n}) \\ - Q_n(\mb{b},x_n^{i-1},\mb{v}_{-n})\big)x_n^i, 
\eeqn
where $Q_n$ is given by \eqref{eq:winning-prob}, and $Q_n(\mb{b},x_n^0,\mb{v}_{-n}) \triangleq 0$. 
\end{enumerate}
\end{algorithm}

\begin{proposition} \label{proposition:MWA-optimality}
The MWA gives a solution of the MOAP.
\end{proposition}
\begin{proof}
Let $(\bs{\pi}^o, \mb{M}^o)$ be the solution given by the MWA and let $[Q^o_n]_{n \in \setN}$ and $[q^o_n]_{n \in \setN}$ be obtained from $\bs{\pi}^o$ by \eqref{eq:winning-prob} and \eqref{eq:q}. Since $\mc{W}(\mb{b},\mb{v}) \subseteq \setA(\mb{b})$, $(\bs{\pi}^o, \mb{M}^o)$ satisfy the FA constraint. 

From Lemma \ref{lemma:monotone-vv-lemma}, $\overline{w}_n(b_n,x_n^i) \leq \overline{w}_n(b_n,x_n^{i+1})$. Hence, for any $(\mb{b}_{-n},\mb{v}_{-n})$, if $A \in \mc{W}(\mb{b},x_n^i,\mb{v}_{-n})$ and $n \in A$, then from step $2$ of the MWA\footnote{The allocation rule $\bs{\pi}^o$ must be consistent in the following sense: let $v_n$ and $\hat{v}_n$ be such that $v_n < \hat{v}_n$, but $\overline{w}_n(b_n,v_n) = \overline{w}_n(b_n,\hat{v}_n)$, then $\Prob{n \in W(\mb{b},v_n,\mb{v}_{-n})} \leq \Prob{n \in W(\mb{b},\hat{v}_n,\mb{v}_{-n})}$ for any $\mb{b}$ and $\mb{v}_{-n}$.}, $A \in \mc{W}(\mb{b},x_n^{i+1},\mb{v}_{-n})$. This in turn implies $Q^o_n(\mb{b},x_n^i,\mb{v}_{-n}) \leq Q^o_n(\mb{b},x_n^{i+1},\mb{v}_{-n})$ and $q_n^o(b_n,x_n^i) \leq q_n^o(b_n,x_n^{i+1})$. Thus, monotonicity condition of Proposition~\ref{proposition:opt-revenue} is satisfied and $\mb{M}^o$ is optimal given~$\bs{\pi}^o$.

Similar to \eqref{eq:winning-prob}, for all $b_n$ and $v_n$, we have:
\beqn
\sum_{n=1}^{N}Q_n^o(\mb{b},\mb{v})\overline{w}_n(b_n,v_n) = \sum_{A \in 2^{\setN}}\pi_A^o(\mb{b},\mb{v})\bigg(\sum_{n \in A}\overline{w}_n(b_n,v_n)\bigg).
\eeqn
Thus, $\bs{\pi}^o$ maximizes $\sum_{n=1}^{N}Q_n(\mb{b},\mb{v})\overline{w}_n(b_n,v_n)$ for each bid vector $(\mb{b},\mb{v})$, subject to the FA constraint.  Proposition \ref{proposition:monotone-vv-opt} then completes the proof.
\end{proof}

The MWA can be interpreted as follows. Given a bid vector $(\mb{b},\mb{v})$, construct a graph $\setG_{\mb{b}}(\setN, \setE)$ with a node $n$ for each buyer $n$, and an edge $e_{n,m} \in \setE$ if $b_n \cap b_m \neq \emptyset$. Thus, $\setG_{\mb{b}}$ is the conflict graph of the buyers, where an edge denotes that buyers corresponding to its endpoints cannot be allocated their bundles simultaneously. The collection of all independent sets of this graph is precisely $\setA(\mb{b})$. Let $\overline{w}_n(b_n,v_n)$ be the weight of node $n$. Then the set of winners $W(\mb{b},\mb{v})$ is a maximum weight independent set of this graph.

In the subsequent discussion, we will be using the MWA with a \textit{deterministic tie-breaking rule} (henceforth deterministic MWA). Here, in step $2$ of the MWA, the set of winners $W(\mb{b},\mb{v}) \in \mc{W}(\mb{b},\mb{v})$ is selected by a deterministic rule. For example, the set of winners $W(\mb{b},\mb{v})$ can be the first allocation set in $\mc{W}(\mb{b},\mb{v})$ under the lexicographic order defined over the set of all allocations $2^{\setN}$.

Let $(\bs{\pi}^o, \mb{M}^o)$ be the solution given by the deterministic MWA. Then $Q^o_n(\mb{b},\mb{v}) \in \{0,1\}$ for all $n \in \setN$, $\mb{b}$, and $\mb{v}$. Also, from the proof of Proposition \ref{proposition:MWA-optimality}, $Q^o_n(\mb{b},\mb{v})$ is nondecreasing in $v_n$, keeping~$\mb{b}$ and $\mb{v}_{-n}$ constant. This, along with the payment rule in step $3$ of MWA, implies that a winner pays the price that is the minimum value he needs to report to still win, keeping his bundle and the bids of everyone else fixed.

\subsection{Solution of the OAP} \label{sec:oap-solution}
We now give a sufficient condition under which a solution of MOAP is also the solution of OAP. To this end, define the \textit{hazard rate ordering} \cite{Shaked&Shanthikumar06} on two random variables as follows: 
\begin{definition} \label{assumption:hazard-rate-order}
A nonnegative random variable $Z_1$ is said to be smaller than a nonnegative random variable $Z_2$ under the hazard rate order, denoted by $Z_1 \leq_h Z_2$, if $Z_1$ and $Z_2$ have the same support, and
\beq{hazard-rate-eq}
\Prob{Z_1 > z | Z_1 > \hat{z}} \leq \Prob{Z_2 > z | Z_2 > \hat{z}},
\eeq
for all $z,\hat{z}$ such that $z \geq \hat{z}$, and $z,\hat{z}$ are in the common support of $Z_1$ and $Z_2$. 
\end{definition}
Notice that, if $Z_1 \leq_h Z_2$, then $Z_1$ is also smaller than $Z_2$ under the first order stochastic dominance (FOSD) \cite{Shaked&Shanthikumar06}. In the FOSD, \eqref{eq:hazard-rate-eq} is replaced by simply $\Prob{Z_1 > z} \leq \Prob{Z_2 > z}$ for all $z$ in the common support of $Z_1$ and $Z_2$. Hence, the hazard rate order is stricter than the FOSD.

It is natural to expect that if there are two bundles where one contains the other, then the larger bundle is likely to have a higher value. This is precisely captured by Assumption \ref{assumption:hazard-order} below.
\begin{assumption} \label{assumption:hazard-order}
For any $s,t \in \setB_n$ with $s \subseteq t$, the conditional random variable $(X_n|B_n=s)$ is smaller than the conditional random variable $(X_n|B_n=t)$ under the hazard rate order. Equivalently, for all $n \in \setN$, $s,t \in \setB_n$ such that $s \subseteq t$, and $1\leq j \leq i \leq K_n$,
\beq{hazard-rate-assumption}
\frac{\sum_{l=i}^{K_n}\pv{x_n^l}{s}}{\sum_{l=j}^{K_n}\pv{x_n^l}{s}} \leq \frac{\sum_{l=i}^{K_n}\pv{x_n^l}{t}}{\sum_{l=j}^{K_n}\pv{x_n^l}{t}}.
\eeq
\end{assumption}

Propositions \ref{proposition:mvv-order} and \ref{proposition:OAP-solution} below describe the main results of this paper. 

\begin{proposition} \label{proposition:mvv-order}
Let $s,t \in \setB_n$ be such that $s \subseteq t$. Then under Assumption \ref{assumption:hazard-order}, $\overline{w}_n(s,x_n^i) \geq \overline{w}_n(t,x_n^i)$ for $1 \leq i \leq K_n$.
\end{proposition}
\begin{proof}
The proof is given in Appendix \ref{sec:appendixC}.
\end{proof}

\begin{proposition} \label{proposition:OAP-solution}
\textit{
A deterministic MWA gives a solution of the OAP under Assumption \ref{assumption:hazard-order}.
}
\end{proposition}
\begin{proof}
Notice that OAP and MOAP differ only in their constraints, and the relaxed IC constraint \eqref{eq:relaxed-ic} is a subset of the IC constraint \eqref{eq:ic}. Hence, we only need to verify that the solution given by the deterministic MWA satisfies the IC constraint. We show that, under the deterministic MWA, the truthful declaration of the types is a weakly dominant strategy for the buyers. 

Let the bid vector be $(\mb{b},\mb{v})$. Based on the reported bundles~$\mb{b}$, the conflict graph $\setG_{\mb{b}}$ is constructed. The weights of the nodes of $\setG_{\mb{b}}$ are the MVVs for the bid vector $(\mb{b},\mb{v})$. Consider a buyer~$n$. Let his true type be $(b_n^*, v_n^*)$. Since buyers are single minded, it can be assumed that $b_n \supseteq b_n^*$, otherwise the payoff from misreporting a bundle can be at most zero, which is less than or equal to the payoff from reporting the bundle truthfully. Also, if buyer $n$ does not get his bundle by bidding $(b_n,v_n)$ (and hence payoff equal to zero) then truthful bidding (payoff at least zero) cannot be worse. Hence, we only need to analyze the case where buyer~$n$ wins by bidding $(b_n, v_n)$ such that $b_n \supseteq b_n^*$. Since buyer $n$ is a winner, there is a maximum weight independent set (henceforth MWIS) in $\setG_{\mb{b}}$ that contains node~$n$. Because of a deterministic tie-breaking rule,~buyer $n$ pays the minimum value he needs to report to win. This is his value $x_n^i$ at which the value of the MWIS containing node~$n$ exceeds the value of all MWIS not containing node~$n$. Now, if instead buyer~$n$ reports $b_n^*$, it can result in deletion of some edges incident on node $n$ in $\setG_{\mb{b}}$, but cannot add any new edge. At the same time, from Proposition \ref{proposition:mvv-order}, the weight of node~$n$ (or the MVV of buyer~$n$) can possibly increase but cannot decrease. Hence, the value of the MWIS containing node $n$ can possibly go up but cannot decrease, while the value of the MWIS not containing node $n$ does not change. Buyer~$n$ still wins and the payment if he declares $(b_n^*,v_n)$ cannot be more than what he pays when he declares $(b_n,v_n)$. Thus, truthful reporting of the bundle is a weakly dominant strategy. 

We can now assume that buyer $n$ reports his bundle $b_n^*$ truthfully. Since the price that he pays only depends on his reported bundle and the bids of everyone else, but not on his reported value, truthful reporting of the value is a weakly dominant strategy. This completes the proof. 
\end{proof}

\section{Discussion} \label{sec:discussion}
In this section, we highlight some of the important properties of the optimal auction characterized in Section \ref{sec:opt-auction}.

\begin{enumerate}[(a)]
\item
\textbf{Some special cases}: A special case of interest is when the bundle that a buyer is interested in is known to everyone. This is equivalent to $\setB_n$ containing only one bundle. Here, OAP and MOAP are identical. More generally, if $\setB_n$ is such that no bundle in $\setB_n$ is a superset of another bundle in $\setB_n$, then a buyer will report his bundle truthfully, and hence OAP and MOAP are identical. If $B_n$ and $X_n$ are independent then Assumption \ref{assumption:hazard-rate-order} trivially holds true. In all these cases, an optimal auction is given by the deterministic MWA. 

Auctions with identical items can be thought of as single minded buyers with \textit{substitutes} and is closely related to the case with known bundles. Here, the seller has $\kappa$ identical items for sale. A buyer is interested in any one of the $\kappa$ items. The analysis of Section \ref{sec:opt-auction} is easily extended to this case by simply defining the collection of feasible allocations as:
\beqn 
\setA_{\kappa} \triangleq \{A: A \subseteq \setN, |A| \leq \kappa \}.
\eeqn
This is the collection of all subsets of $\setN$ with cardinality less than or equal to $\kappa$. The optimal auction is given by the MWA with $\setA(\mb{b})$ replaced by~$\setA_{\kappa}$.

\item
\textbf{Continuous OAP}: The results of this paper easily extend to the continuous version of the OAP where buyers' valuation sets are continuous. Here, $\setX_n$ is a nonnegative interval of $\R$, and the random variable $X_n$ is specified by the probability density function (pdf) $f_{X_n}(x) > 0$ for all $x \in \setX_n$. Let $F_{X_n}(x)$ be the corresponding cumulative distribution (CDF) function. Denote the conditional pdf and conditional CDF of $X_n$, given $B_n = b_n$, by $f_{X_n|b_n}(x)$ and $F_{X_n|b_n}(x)$ respectively. Then the MWA again gives the solution of the continuous OAP after the following modifications: 
\begin{align*}
w_n(b_n,v_n) &= v_n - \frac{1-F_{X_n|b_n}(v_n)}{f_{X_n|b_n}(v_n)}, \\
M_n(\mb{b},\mb{v}) &= Q_n(\mb{b},\mb{v})v_n - \int_0^{v_n}Q_n(\mb{b},x,\mb{v}_{-n})dx.
\end{align*}

The continuous OAP, however, has one key difference from the discrete OAP. In the continuous case, given an allocation rule $\bs{\pi}$, under the IC and relaxed IR constraints, the expected payment that a buyer makes is determined up to an additive constant (the \textit{revenue equivalence principle} of \cite{Myerson81}); i.e., $m_n(b_n,v_n) - m_n(b_n,0)$ is a known function of $v_n$. However, in the discrete case, $m_n(b_n,x_n^i) - m_n(b_n,0)$ is fixed to within an interval of values. The expected revenue is maximized by taking the upper value of this interval, and the optimal auction in Section \ref{sec:opt-auction} is characterized this way.

\item
\textbf{Reserve prices}: Given a bid vector $(\mb{b},\mb{v})$, no buyer $n$ with $\overline{w}_n(b_n,v_n) < 0$ is included in the set of winners, $W(\mb{b},\mb{v})$, by the MWA. Depending on the tie-breaking rule, a buyer $n$ with $\overline{w}_n(b_n,v_n) = 0$ may or may not be included in the set of winners. Assume that only buyers with $\overline{w}_n(b_n,v_n) > 0$ are considered. Since $\overline{w}_n(b_n,x_n^i) \leq \overline{w}_n(b_n,x_n^{i+1})$, the seller equivalently sets reserve prices for each buyer $n$ and does not sell any item to a buyer if his reported value is less than his reserve price. The reserve price for a buyer depends only on the probability distribution of his valuations, conditioned on his bundle, and might be different for different buyers. Let $r_n(b_n)$ be the reserve price for a buyer $n$ as a function of his bundle. If $x_n^k = r_n(b_n)$, then $\overline{w}_n(b_n,x_n^i) \leq 0$ for $1\leq i \leq k-1$. As $\overline{w}_n(b_n,x^i)$'s are the slopes of the lines joining the points $(g_n^{b_n,i},\overline{h}_n^{b_n,i})$, we get $\overline{h}_n^{b_n,k-1} = \min_{0 \leq i \leq K_n}\overline{h}_n^{b_n,i}$. From the property of convex hull, $\min_{0 \leq i \leq K_n}\overline{h}_n^{b_n,i} = \min_{0 \leq i \leq K_n} h_n^{b_n,i}$. Thus, using the definition of $h_n^{b_n,i}$'s, an equivalent formulation of the reserve price is:
\beq{reserve-price-eqiv}
r_n(b_n) = \max \bigg\{v_n: ~ v_n \in \argmax_{\hat{v}_n \in \setX_n} ~\hat{v}_n\Prob{X_n \geq \hat{v}_n|B_n = b_n}\bigg\}.
\eeq
Graphically, this corresponds to the y-intercept of the line through the lowermost point of the graph and the point $(1,0)$ ($x^3$ in Figures \ref{fig:virtual-bids} and~\ref{fig:monotone-bids}).

\item
\textbf{Implementation complexity}: The optimal allocation rule for auctions with single-minded buyers requires finding a maximum weight independent set in the conflict graph. This problem is NP-hard. However, similar to \cite{Lehmann02}, a greedy scheme can be obtained that is easy to implement, and achieves $\sqrt{S}$ approximation\footnote{Any approximation better than $\sqrt{S}$ is again NP-Hard.} of the revenue generated by the deterministic MWA. The greedy scheme allocates the bundles to the buyers according to the order induced by the normalized virtual valuations $w_n(b_n,v_n)/\sqrt{|b_n|}$. The price charged to a buyer who gets his desired bundle is the minimum value he needs to report to still win, keeping his bundle and the bids of everyone else fixed.

\item
\textbf{On the hazard rate order assumption}: In the absence of Assumption \ref{assumption:hazard-order}, the solution given by the MWA (under any tie-breaking rule) need not satisfy the IC constraint. The following example shows this. Consider two buyers $\{1,2\}$ and two items $\{A,B\}$. Buyer $1$ is interested in bundle $\{A\}$ and has value $\$1$ for it. Buyer $2$ can be interested in bundle $\{A\}$ or bundle $\{A,B\}$, each with probability $1/2$. Conditioned on buyer $2$ being interested in bundle $\{A\}$, his values can be $\$2$ or \$$4$, each with probability $1/2$. Conditioned on him being interested in bundle $\{A,B\}$, his values can be $\$2$ or $\$4$, with probabilities $0.9$ and $0.1$ respectively. Clearly, Assumption \ref{assumption:hazard-order} does not hold true for buyer $2$. The virtual-valuation function for buyer $1$ is $w_1(\{A\},\$1) = \$1$. For buyer $2$, the virtual-valuation function is $w_2(\{A\},\$2) = \$0$, $w_2(\{A\},\$4) = \$4$, $w_2(\{A,B\},\$2) = \$16/9$, and $w_2(\{A,B\},\$2) = \$4$. Under the MWA, if buyer $2$ bids $(\{A\},\$2)$ he loses, and if he bids $(\{A\},\$4)$ then he gets bundle $\{A\}$ at the price $\$4$. However, if buyer $2$ bids $(\{A,B\}, \$2)$ or $(\{A,B\}, \$4)$ then he gets bundle $\{A,B\}$ at the price $\$2$. Thus, if the true type of buyer $2$ is $(\{A\},\$4)$, he will misreport it to $(\{A,B\}, \$2)$ or $(\{A,B\}, \$4)$.
\end{enumerate}

\section{Conclusions} \label{sec:conclusions}
We characterized a Bayesian revenue optimal multiple items auction with single-minded buyers under a partial hazard rate order assumption on the conditional distribution of any buyer's valuation. This assumption is intuitive for single-minded buyers and imply that the larger bundle is likely to have higher value. The resulting auction has a simple structure - the set of winners are the maximum weight independent set of the conflict graph of the buyers, and the payment made by a winner is the minimum value he needs to report to win. Single-minded buyers have two dimensional private information. The contributions here provide a step towards understanding optimal auction problems where buyers' private information is multidimensional.

\appendix
\section{Proof of Proposition \ref{proposition:opt-revenue}} \label{sec:appendixA}
The proof of Proposition \ref{proposition:opt-revenue} follows from the lemmas given below.

\begin{lemma}
\label{lemma:monotone_lemma}
Under the relaxed IC constraint \eqref{eq:relaxed-ic}, $q_n(b_n,x_n^i) \leq q_n(b_n,x_n^{i+1})$ for all $n \in \setN$, $b_n$, and $1 \leq i \leq K_n-1$.
\end{lemma}
\begin{proof}
The proof follows easily from \eqref{eq:relaxed-ic} by considering the case where the true value of the bundle $b_n$ for buyer $n$ is $x_n^i$ but he reports $x_n^{i+1}$ instead, and the case where the true value is $x_n^{i+1}$ but he reports $x_n^i$ instead.
\end{proof}

\begin{lemma}
\label{lemma:ic_reduction_lemma}
The relaxed IC constraint \eqref{eq:relaxed-ic} is equivalent to:
\beq{ic_reduction}
\big(q_n(b_n,x_n^{i+1}) - q_n(b_n,x_n^i)\big)x_n^i \leq m_n(b_n,x_n^{i+1}) - m_n(b_n,x_n^i)
\leq \big(q_n(b_n,x_n^{i+1}) - q_n(b_n,x_n^i)\big)x_n^{i+1}, 
\eeq
for all $n \in \setN$, $b_n$, and $1 \leq i \leq K_n-1$.
\end{lemma}
\begin{proof}
Trivially, the relaxed IC constraint \eqref{eq:relaxed-ic} implies \eqref{eq:ic_reduction}. To show that \eqref{eq:ic_reduction} implies \eqref{eq:ic}, first consider the case $j > i$.  Using \eqref{eq:ic_reduction}, 
\begin{align*}
\big(q_n(b_n,x_n^j) - q_n(b_n,x_n^i)\big)x_n^i & \leq \sum_{k=i}^{j-1}\left[\big(q_n(b_n,x_n^{k+1}) - q_n(b_n,x_n^k)\big)x_n^k \right], \\
& \leq \sum_{k = i}^{j-1}\left[m_n(b_n,x_n^{k+1}) - m_n(b_n,x_n^k) \right], \\ 
& = m_n(b_n,x_n^j) - m_n(b_n,x_n^i),
\end{align*}
where the first inequality follows from Lemma \ref{lemma:monotone_lemma} and $x_n^k < x_n^{k+1}$. Thus, \eqref{eq:relaxed-ic} holds for $j < i$. Similarly, starting with the left inequality of \eqref{eq:ic_reduction}, it can easily be shown that \eqref{eq:ic_reduction} implies \eqref{eq:relaxed-ic} for $ j < i$, and the proof is complete.
\end{proof}

\begin{lemma}
\label{lemma:characterizing_lemma}
Under the relaxed IC constraint \eqref{eq:relaxed-ic} and the IR constraint \eqref{eq:ir}, for all $n \in \setN$ and $b_n$, the following holds:
\beq{paymenb_upper_bound}
\E{m_n(b_n,X_n)|B_n = b_n} \leq \E{q_n(b_n,X_n)w_n(b_n,X_n)|B_n = b_n}.
\eeq
Moreover, \eqref{eq:paymenb_upper_bound} holds with equality for $m_n(b_n,x_n^i)$ satisfying:
\beq{m-equality}
m_n(b_n,x_n^i) = \sum_{j = 1}^{i}\left[\big(q_n(b_n,x_n^j) - q_n(b_n,x_n^{j-1})\big)x_n^j\right],
\eeq
where we use the notational convention $q_n(b_n,x_n^{0}) \triangleq 0$.
\end{lemma}
\begin{proof}
Lemma \ref{lemma:ic_reduction_lemma} and the IR constraint $m_n(b_n,x_n^1) \leq q_n(b_n,x_n^1)x_n^1$ easily imply:
\beq{ic_reduction_lemma_eq1}
m_n(b_n,x_n^i) \leq \sum_{j=1}^{i}\left[\big(q_n(b_n,x_n^j) - q_n(b_n,x_n^{j-1})\big)x_n^j\right],
\eeq
where $q_n(b_n,x_n^{0}) \triangleq 0$. Using \eqref{eq:ic_reduction_lemma_eq1}, 
\begin{align*}
\E{m_n(b_n,X_n)|B_n = b_n} & = \sum_{i=1}^{K_n} \pv{x_n^i}{b_n} m_n(b_n,x_n^i), \\ 
& \leq \sum_{i=1}^{K_n}\sum_{j=1}^{i}\left[\big(q_n(b_n,x_n^j) - q_n(b_n,x_n^{j-1})\big)x_n^j \pv{x_n^i}{b_n}\right], \\
& = \sum_{j=1}^{K_n}\sum_{i=j}^{K_n}\left[\big(q_n(b_n,x_n^j) - q_n(b_n,x_n^{j-1})\big)x_n^j \pv{x_n^i}{b_n}\right], \\
& = \sum_{j=1}^{K_n}\left[\big(q_n(b_n,x_n^j) - q_n(b_n,x_n^{j-1})\big)\left(\sum_{i=j}^{K_n} \pv{x_n^i}{b_n} \right)x_n^j\right], \\
& = \sum_{j=1}^{K_n} \pv{x_n^j}{b_n} q_n(b_n,x_n^j) w_n(b_n,x_n^j), \\ 
& = \E{q_n(b_n,X_n)w_n(b_n,X_n)|B_n = b_n}.
\end{align*}
where the second last equality is obtained by rearranging the terms and using \eqref{eq:virtual-bid}. It is straightforward to verify that the above holds with equality for $m_n(b_n,x_n^i)$ given by \eqref{eq:m-equality}. The final step is show that this particular choice of $m_n$ satisfies the relaxed IC and the IR constraints. The relaxed IC constraint is trivially satisfied using Lemma~\ref{lemma:ic_reduction_lemma}. The IR constraint is satisfied since:
\begin{align*}
m_n(b_n,x_n^i) & = \sum_{j = 1}^{i}(q_n(b_n,x_n^j) - q_n(b_n,x_n^{j-1}))x_n^j, \\ 
& \leq \sum_{j = 1}^{i}\big(q_n(b_n,x_n^j) - q_n(b_n,x_n^{j-1})\big)x_n^i = q_n(b_n,x_n^i)x_n^i.
\end{align*}
\end{proof}
Notice that the last part of the proof of Lemma \ref{lemma:characterizing_lemma} shows that the condition $q_n(b_n,x_n^i) \leq q_n(b_n,x_n^{i+1})$ for $1 \leq i \leq K_n-1$, is also sufficient for the existence of a payment rule satisfying the relaxed IC and the IR constraint. Proposition \ref{proposition:opt-revenue} easily follows by combining Lemmas \ref{lemma:monotone_lemma}-\ref{lemma:characterizing_lemma}, and noticing that:
\beqn
\E{M_n(\mb{Y})} = \E{m_n(Y_n)} = \E{\E{m_n(B_n,X_n)|B_n}} = \sum_{b_n \in \setB_n}\pb{b_n}\E{m_n(b_n,X_n)|B_n = b_n}.
\eeqn

\section{Proof of Proposition \ref{proposition:monotone-vv-opt}} \label{sec:appendixB}
To prove \eqref{eq:monotone_bid_ineq}, it is sufficient to show that:
\beq{monotone_bid_opb_lemma_eq1}
\E{q_n(Y_n)w_n(Y_n)} \leq \E{q_n(Y_n)\overline{w}_n(Y_n)}.
\eeq

We have:
\beq{monotone_bid_opb_lemma_eq1A}
\E{q_n(Y_n)w_n(Y_n)} = \sum_{b_n \in \setB_n}\pb{b_n}\E{q_n(b_n,X_n)w_n(b_n,X_n)| B_n = b_n}.
\eeq
Using \eqref{eq:h-slope}, we can write:
\begin{align} \label{eq:monotone_bid_opb_lemma_eq2}
& \E{q_n(b_n,X_n)w_n(b_n,X_n) | B_n = b_n} \nonumber \\
& = \sum_{i=1}^{K_n} \pv{x_n^i}{b_n} q_n(b_n,x_n^i) w_n(b_n,x_n^i), \nonumber \\
& = \sum_{i=1}^{K_n} \pv{x_n^i}{b_n} q_n(b_n,x_n^i) \left(\frac{h_n^{b_n,i}-h_n^{b_n,i-1}}{g_n^{b_n,i}-g_n^{b_n,i-1}}\right), \nonumber \\
& = \sum_{i=1}^{K_n} q_n(b_n,x_n^i)(h_n^{b_n,i}-h_n^{b_n,i-1}), \nonumber \\
& = h_n^{b_n,K_n} q_n(b_n,x_n^{K_n}) -h_n^{b_n,0} q_n(b_n,x_n^1) - \sum_{i=1}^{K_n-1}h_n^{b_n,i}(q_n(b_n,x_n^{i+1})-q_n(b_n,x_n^i)).
\end{align}

Similarly, 
\begin{multline} \label{eq:monotone_bid_opb_lemma_eq3}
\E{q_n(b_n,X_n)\overline{w}_n(b_n,X_n) | B_n = b_n} \\ 
= \overline{h}_n^{b_n,K_n} q_n(b_n,x_n^{K_n}) - \overline{h}_n^{b_n,0} q_n(b_n,x_n^1) - \sum_{i=1}^{K_n-1}\overline{h}_n^{b_n,i}(q_n(b_n,x_n^{i+1})-q_n(b_n,x_n^i)).
\end{multline}

Since $\overline{h}_n^{b_n,i}$ is the point corresponding to $g_n^{b_n,i}$ on the convex hull of $[(g_n^{b_n,i},h_n^{b_n,i})]_{1 \leq i \leq K_n}$, we must have $h_n^{b_n,0} = \overline{h}_n^{b_n,0}$, $h_n^{b_n,K_n} = \overline{h}_n^{b_n,K_n}$, and $h_n^{b_n,i} \geq \overline{h}_n^{b_n,i}$. This, along with $q_n(b_n,x_n^{i+1}) \geq q_n(b_n,x_n^i)$, and \eqref{eq:monotone_bid_opb_lemma_eq2}-\eqref{eq:monotone_bid_opb_lemma_eq3}, gives:
\begin{multline*}
\E{q_n(b_n,X_n)w_n(b_n,X_n)- q_n(b_n,X_n)\overline{w}_n(b_n,X_n) | B_n = b_n} \\
= -\sum_{i=1}^{K_n-1}(h_n^{b_n,i}-\overline{h}_n^{b_n,i})(q_n(b_n,x_n^{i+1})-q_n(b_n,x_n^i)) \leq 0,
\end{multline*}
hence proving \eqref{eq:monotone_bid_opb_lemma_eq1}, and in turn, the inequality \eqref{eq:monotone_bid_ineq}.

Let $\overline{\bs{\pi}}$ be the allocation rule that maximizes $\sum_{n=1}^{N}Q_n(\mb{b},\mb{v})\overline{w}_n(b_n,v_n)$ for each bid vector $(\mb{b},\mb{v})$, subject to the FA constraint. Given $b_n$, if $0 \leq i < j \leq K_n$ are such that $h_n^{b_n,i} = \overline{h}_n^{b_n,i}$, $h_n^{b_n,k} > \overline{h}_n^{b_n,k}$ for $i+1\leq k \leq j-1$, and $h_n^{b_n,j} = \overline{h}_n^{b_n,j}$ (recall that $h_n^{b_n,i} \geq \overline{h}_n^{b_n,i}$), then $\overline{h}_n^{b_n,k}$ lies on the line joining $(g_n^{b_n,i},h_n^{b_n,i})$ and $(g_n^{b_n,j},h_n^{b_n,j})$. Hence, $\overline{w}_n(b_n,x_n^l) = \overline{w}_n(b_n,x_n^{l+1})$ for $i+1\leq l \leq j$; i.e., $\overline{w}_n$ is constant in this interval. This in turn implies that if $A \in 2^{\setN}$ such that $n \in A$, then $\overline{\pi}_A(\mb{b},x_n^l,\mb{v}_{-n})$ is constant in the interval $i+1\leq l \leq j$, given $\mb{b}$ and $\mb{v}_{-n}$. Let $\overline{Q}_n$ and $\overline{q}_n$ be obtained from $\overline{\bs{\pi}}$ by \eqref{eq:winning-prob} and \eqref{eq:q}. Then $\overline{q}_n(b_n,x_n^l)$ is also constant in the interval $i+1\leq l \leq j$. Moreover, from the construction of $\overline{w}_n$,
\begin{align*}
\sum_{l=i+1}^j \pv{x_n^l}{b_n} w_n(b_n,x_n^l) &= \sum_{l=i+1}^j \pv{x_n^l}{b_n} \overline{w}_n(b_n,x_n^l),\\
\Rightarrow \sum_{l=i+1}^j \pv{x_n^l}{b_n} w_n(b_n,x_n^l) \overline{q}_n(b_n,x_n^l) &= \sum_{l=i+1}^j \pv{x_n^l}{b_n} \overline{w}_n(b_n,x_n^l) \overline{q}_n(b_n,x_n^l).
\end{align*}
Thus, for all $n \in \setN$ and $b_n$, we have:
\begin{align}\label{eq:monotone_bid_opb_lemma_eq4}
\E{w_n(b_n,X_n) \overline{q}_n(b_n,X_n) | B_n = b_n} &= \E{\overline{w}_n(b_n,X_n) \overline{q}_n(b_n,X_n) | B_n = b_n}, \nonumber \\
\Rightarrow \E{\overline{q}_n(Y_n)w_n(Y_n)} &= \E{\overline{q}_n(Y_n)\overline{w}_n(Y_n)},
\end{align}
where the last equality follows from \eqref{eq:monotone_bid_opb_lemma_eq1A}. This completes the proof of the equality.

\section{Proof of Proposition \ref{proposition:mvv-order}} \label{sec:appendixC}
Define $F_{X_n|b_n}(z) \triangleq \Prob{X_n < z |B_n = b_n}$. Notice that this is the \textit{left continuous} cumulative distribution function (CDF) of the conditional random variable $(X_n | B_n = b_n)$. We start with the following lemmas:

\begin{lemma} \label{lemma:monotone-vv-order-lemma0}
For all $n \in \setN$ and $b_n$, $\overline{w}_n(b_n,x_n^i) < x_n^i$ for $1 \leq i \leq K_n -1$, and $\overline{w}_n(b_n,x_n^{K_n}) = x_n^{K_n}$.
\end{lemma}
\begin{proof}
If $\overline{w}_n(b_n,x_n^i) = w_n(b_n,x_n^i)$, then this is trivially true. Given $b_n$, if $0 \leq i < j \leq K_n$ are such that $h_n^{b_n,i} = \overline{h}_n^{b_n,i}$, $h_n^{b_n,k} > \overline{h}_n^{b_n,k}$ for $i+1\leq k \leq j-1$, and $h_n^{b_n,j} = \overline{h}_n^{b_n,j}$, then $\overline{w}_n(b_n,x_n^k)$ is constant in the interval $i+1 \leq k \leq j$. Hence, 
$\overline{w}_n(b_n,x_n^{k}) = \overline{w}_n(b_n,x_n^{i+1}) < w_n(b_n,x_n^{i+1}) < x_n^{i+1} < x_n^k.$ This completes the proof of inequality. Also, $w_n(b_n,x_n^{K_n}) = x_n^{K_n} > w_n(b_n,x_n^i)$, for any $1 \leq i \leq K_n -1$.  Thus, the convex hull of points $[(g_n^{b_n,i},h_n^{b_n,i})]_{1 \leq i \leq K_n}$ always contains the line connecting $(g_n^{b_n,K_n-1},h_n^{b_n,K_n-1})$ to $(g_n^{b_n,K_n},h_n^{b_n,K_n})$ (recall the construction of $\overline{w}_n$ in Section \ref{sec:moap-solution}). Hence, $\overline{w}_n(b_n,x_n^{K_n}) = w_n(b_n,x_n^{K_n}) = x_n^{K_n}$, and the proof is complete.
\end{proof}

\begin{lemma} \label{lemma:monotone-vv-order-lemma1}
For all $n \in \setN$, $b_n$, and $1 \leq i \leq K_n -1$, 
\beq{monotone-vv-order-eq1}
\overline{w}_n(b_n,x_n^i) = \argmin_{c < x_n^i}\max_{z \in \left[x_n^1,x_n^{K_n}\right]}\frac{(z-c)(1-F_{X_n|b_n}(z))}{x_n^i - c},
\eeq
with the convention that if more than one value of $c$ minimizes the maximum, then the largest such~$c$ is selected.
\end{lemma}
\begin{proof}
From \eqref{eq:gh}, for any $1 \leq i \leq K_n$,
\beq{monotone-vv-order-eq2}
(g_n^{b_n,i-1},h_n^{b_n,i-1}) = \big(F_{X_n|b_n}(x_n^i),-x_n^i(1-F_{X_n|b_n}(x_n^i)\big).
\eeq
Let $I(z) \triangleq \big(F_{X_n|b_n}(z),-z(1-F_{X_n|b_n}(z)\big)$ for $z \in [x_n^1,x_n^{K_n}]$. Consider the plot of points $I(z)$'s. The convex hull of the points $[I(z)]_{z \in [x_n^1,x_n^{K_n}]}$ is same as that of the points $[(g_n^{b_n,i},h_n^{b_n,i})]_{0 \leq i \leq K_n-1}$. Note that, from the proof of Lemma \ref{lemma:monotone-vv-order-lemma0}, the convex hulls of the points $[(g_n^{b_n,i},h_n^{b_n,i})]_{0 \leq i \leq K_n-1}$ and $[(g_n^{b_n,i},h_n^{b_n,i})]_{0 \leq i \leq K_n}$ differ from each other by just the line segment connecting $(g_n^{b_n,K_n-1},h_n^{b_n,K_n-1})$ to $(g_n^{b_n,K_n},h_n^{b_n,K_n})$. Hence, $\overline{w}_n(b_n,x_n^i)$'s for $1 \leq i \leq K_n -1$ are obtained as the slopes of the convex hull of points $I(z)$'s. Fix $x_n^i$ for some $1 \leq i \leq K_n -1$. Call the line from $(0,-x_n^i)$ to $(1,0)$ {\em the line for $x_n^i$}. Given $z$ and $c < x_n^i$, consider the line through the point $I(z)$ with slope $c$, and let $J$ be the point of intersection of this line with the line for $x_n^i$. Then, $(z-c)(1-F_{X_n|b_n}(z))/(x_n^i - c)$ is the horizontal distance of $J$ from the vertical line at $(1,0)$. Taking the maximum over $z$ corresponds to the point $J_c$ which is the intersection of the line of slope $c$ that is tangent to the plot, and the line for $x_n^i$. Then the minimizing $c$ is the slope of the tangent at the point $J^*$ where the convex hull of $I(z)$ intersects the line for $x_n^i$, and hence, same as $\overline{w}_n(b_n,x_n^i)$. If there is more than one intersection point, the largest is selected. From Lemma \ref{lemma:monotone-vv-order-lemma0}, the minimum is achieved by $c < x_n ^i$. This completes the proof.
\end{proof}

For $c < x_n^{K_n}$, define:
\beqn
\Phi_{X_n|b_n}(c) \triangleq \max_{z \in \left[x_n^1,x_n^{K_n}\right]}(z-c)(1-F_{X_n|b_n}(z)).
\eeqn
Notice that $\Phi_{X_n|b_n}(c)$ is nonincreasing in $c$.

\begin{lemma} \label{lemma:monotone-vv-order-lemma2}
Let $s,t \in \setB_n$ be such that $s \subseteq t$. Then under Assumption \ref{assumption:hazard-order}, $\Phi_{X_n|s}(c)/\Phi_{X_n|t}(c)$ is nonincreasing in $c$. 
\end{lemma}
\begin{proof}
Fix $c_1 < c_2 < x_n^{K_n}$. We need to prove that:
\beqn
\frac{\Phi_{X_n|s}(c_1)}{\Phi_{X_n|t}(c_1)}\leq \frac{\Phi_{X_n|s}(c_2)}{\Phi_{X_n|t}(c_2)}.
\eeqn 
Under Assumption \ref{assumption:hazard-order}, $(1-F_{X_n|s}(z))/(1-F_{X_n|t}(z))$ is nonincreasing in $z$. Let $z_1^{t}$ and $z_2^{s}$ denote the values of $z$ achieving the maximum in the definition of $\Phi_{X_n|t}(c_1)$ and $\Phi_{X_n|s}(c_2)$ respectively. Clearly, $z_1^{t} > c_1$ and $z_2^{s} > c_2$. If $z_1^{t} \leq z_2^{s}$, 
\begin{multline*}
\frac{\Phi_{X_n|s}(c_1)}{\Phi_{X_n|t}(c_1)} \geq \frac{(z_1^t-c_1)(1-F_{X_n|s}(z_1^t))}{(z_1^t-c_1)(1-F_{X_n|t}(z_1^t))} 
= \frac{1-F_{X_n|s}(z_1^t)}{1-F_{X_n|t}(z_1^t)} \\ 
\geq \frac{1-F_{X_n|s}(z_2^s)}{1-F_{X_n|t}(z_2^s)} = \frac{(z_2^s-c_2)(1-F_{X_n|s}(z_2^s))}{(z_2^s-c_2)(1-F_{X_n|t}(z_2^s))}
\geq \frac{\Phi_{X_n|s}(c_2)}{\Phi_{X_n|t}(c_2)},
\end{multline*}

On the other hand,  if $z_1^t \geq z_2^s$,
\beqn
\frac{\Phi_{X_n|s}(c_1)}{\Phi_{X_n|t}(c_1)}
\geq \frac{(z_2^s-c_1)(1-F_{X_n|s}(z_2^s))}{(z_1^t-c_1)(1-F_{X_n|t}(z_1^t))} 
\geq \frac{(z_2^s-c_2)(1-F_{X_n|s}(z_2^s))}{(z_1^t-c_2)(1-F_{X_n|t}(z_1^t))}
\geq \frac{\Phi_{X_n|s}(c_2)}{\Phi_{X_n|t}(c_2)}.
\eeqn
In either case, the required inequality is proved.
\end{proof}

Combining Lemma \ref{lemma:monotone-vv-order-lemma1} and Lemma \ref{lemma:monotone-vv-order-lemma2}, for any $c \leq \overline{w}_n(t,x_n^i)$, and $1 \leq i \leq K_n -1$, 
\beqn
\frac{\Phi_{X_n|s}(c)}{x_n^i-c}  \geq
\frac{\Phi_{X_n|t}(c) \Phi_{X_n|s}(\overline{w}_n(t,x_n^i))}{(x_n^i-c)  \Phi_{X_n|t}(\overline{w}_n(t,x_n^i))} 
\geq \frac{\Phi_{X_n|t}(\overline{w}(t,x_n^i) ) \Phi_{X_n|s}(\overline{w}(t,x_n^i))}{(x_n^i-\overline{w}(t,x_n^i))  \Phi_{X_n|t}(\overline{w}(t,x_n^i))}
= \frac{\Phi_{X_n|s}(\overline{w}(t,x_n^i))}{x_n^i-\overline{w}(t,x_n^i)},
\eeqn
where the first inequality is by Lemma \ref{lemma:monotone-vv-order-lemma2}, and the second is because $\Phi_{X_n|b_n}(c)$ is nonincreasing in $c$. Hence, from Lemma \ref{lemma:monotone-vv-order-lemma1}, it follows that $\overline{w}_n(s,x_n^i) \geq \overline{w}_n(t,x_n^i)$ for $1\leq i  \leq K_n-1$. Also $\overline{w}_n(s,x_n^{K_n}) = \overline{w}_n(t,x_n^{K_n}) = x_n^{K_n}$. This completes the proof of Proposition \ref{proposition:mvv-order}.
 
\bibliographystyle{ieeetr}
\bibliography{AuctionTheory}

\begin{thebibliography}{10}

\bibitem{Clarke71}
E.~H. Clarke, ``Multipart pricing of public goods,'' {\em Public Choice},
  vol.~V11, no.~1, pp.~17--33, 1971.

\bibitem{Groves73}
T.~Groves, ``Incentives in teams,'' {\em Econometrica}, vol.~41, no.~4,
  pp.~617--631, 1973.

\bibitem{Vickrey61}
W.~Vickrey, ``Counterspeculation, auctions, and competitive sealed tenders,''
  {\em The Journal of Finance}, vol.~16, no.~1, pp.~8--37, 1961.

\bibitem{Myerson81}
R.~Myerson, ``Optimal auction design,'' {\em Mathematics of Operations
  Research}, vol.~6, no.~1, pp.~58--73, 1981.

\bibitem{Hartline&Karlin07}
J.~D. Hartline and A.~R. Karlin, ``Profit maximization in mechanism design,''
  in {\em Algorithmic Game Theory} (N.~Nisan, T.~Roughgarden, E.~Tardos, and
  V.~V. Vazirani, eds.), ch.~13, pp.~331--361, New York, NY, USA: Cambridge
  University Press, 2007.

\bibitem{Armstrong2000}
M.~Armstrong, ``Optimal multi-object auctions,'' {\em The Review of Economic
  Studies}, vol.~67, no.~3, pp.~455--481, 2000.

\bibitem{Lehmann02}
D.~Lehmann, L.~I. O\'{c}allaghan, and Y.~Shoham, ``Truth revelation in
  approximately efficient combinatorial auctions,'' {\em J. ACM}, vol.~49,
  no.~5, pp.~577--602, 2002.

\bibitem{Ledyard07}
J.~O. Ledyard, ``Optimal combinatoric auctions with single-minded bidders,'' in
  {\em EC '07: Proceedings of the 8th ACM conference on Electronic commerce},
  (New York, NY, USA), pp.~237--242, ACM, 2007.

\bibitem{Elkind07}
E.~Elkind, ``Designing and learning optimal finite support auctions,'' in {\em
  SODA '07: Proceedings of the eighteenth annual ACM-SIAM symposium on Discrete
  algorithms}, (Philadelphia, PA, USA), pp.~736--745, Society for Industrial
  and Applied Mathematics, 2007.

\bibitem{Dizdar09}
A.~G. Deniz~Dizdar and B.~Moldovanu, ``Revenue maximization in the dynamic
  knapsack problem.'' Working paper, University of Bonn, 2009.

\bibitem{Shaked&Shanthikumar06}
M.~Shaked and J.~G. Shanthikumar, {\em Stochastic Orders (Springer Series in
  Statistics)}.
\newblock New York: Springer-Verlag, October 2006.

\end{thebibliography}


\end{document}